\documentclass[letterpaper, 10 pt, conference]{ieeeconf}  

\IEEEoverridecommandlockouts                              
\usepackage{amsmath,amssymb,amsfonts,amsthm,color,float}
\usepackage{algorithmic}
\usepackage{graphicx}
\usepackage{textcomp}
 \usepackage{calc}
 \usepackage{tikz}
\usepackage{cases}
\usepackage{verbatim}
\usepackage{hyperref}
\usepackage{xcolor, soul}
\usepackage{nccmath}
\usepackage{bbm}
\usepackage{mathtools, nccmath}
\usepackage{cleveref}
\let\labelindent\relax
\usepackage{enumitem}

\usetikzlibrary{arrows}
\usetikzlibrary{shapes}
 \usetikzlibrary{calc}
\usetikzlibrary{decorations.pathreplacing}
\usetikzlibrary{arrows,positioning}

\theoremstyle{plain}
\newtheorem{thm}{Theorem}
\newtheorem{cor}{Corollary}
\newtheorem{prop}{Proposition}

\theoremstyle{definition}

\newtheorem{ex}{Example}

\theoremstyle{remark}

\newcommand{\N}{\mathbb{N}}  
\newcommand{\R}{\mathbb{R}}  
\newcommand{\gm}{\mathrm{G}}  
\newcommand{\Gm}{\mathcal{G}}  
\newcommand{\ps}{\mathcal{N}}  
\newcommand{\rs}{\mathcal{R}}  
\newcommand{\ac}{\mathcal{A}}  
\newcommand{\wf}{\mathrm{W}}  
\newcommand{\ui}{\mathrm{U}_i}  
\newcommand{\ur}{\mathrm{U}_r}  
\newcommand{\x}{\mathrm{Y}}  
\newcommand{\xtr}{y_{\mathrm{true}}} 
\newcommand{\poa}{\mathrm{PoA}}  
\newcommand{\ns}{\mathrm{ne}}  
\newcommand{\opt}{\mathrm{opt}}  
\newcommand*{\Perm}[2]{{}^{#1}\!P_{#2}}  
\newcommand{\cfdel}{B_{\delta}}  
\newcommand{\z}{\mathrm{Z}}  
\newcommand{\Z}{\mathcal{Z}}  
\newcommand{\dtr}{\delta_{\mathrm{true}}}  
\newcommand{\opu}{u^{\rm{opt}}_\delta}  
\newcommand{\cfdeltr}{B_{\dtr}}  

\newcommand{\mc}[1]{\mathcal{#1}}

\renewcommand{\baselinestretch}{0.97}

\begin{document}
\title{Mission Level Uncertainty in Multi-Agent Resource Allocation}

\author{Rohit Konda \and Rahul Chandan \and Jason R. Marden
\thanks{R. Konda (\texttt{rkonda@ucsb.edu}), R. Chandan, and J. R. Marden are with the Department of Electrical and Computer Engineering at the University of California, Santa Barbara, CA. This work is supported by \texttt{ONR Grant \#N00014-20-1-2359} and \texttt{AFOSR Grant \#FA9550-20-1-0054}.}}

\maketitle
\thispagestyle{empty}


\begin{abstract}
In recent years, a significant research effort has been devoted to the design of distributed protocols for the control of multi-agent systems, as the scale and limited communication bandwidth characteristic of such systems render centralized control impossible. Given the strict operating conditions, it is unlikely that every agent in a multi-agent system will have local information that is consistent with the true system state. Yet, the majority of works in the literature assume that agents share perfect knowledge of their environment. This paper focuses on understanding the impact that inconsistencies in agents' local information can have on the performance of multi-agent systems. More specifically, we consider the design of multi-agent operations under a game theoretic lens where individual agents are assigned utilities that guide their local decision making. We provide a tractable procedure for designing utilities that optimize the efficiency of the resulting collective behavior (i.e., price of anarchy) for classes of set covering games where the extent of the information inconsistencies is known. In the setting where the extent of the informational inconsistencies is not known, we show -- perhaps surprisingly -- that underestimating the level of uncertainty leads to better price of anarchy than overestimating it.
\end{abstract}


\section{INTRODUCTION}

The widespread utilization of multi-agent architectures has led to increased interest in the design of distributed control algorithms. Advancements in this area could have significant societal impact, not limited to the detection of forest fires \cite{hefeeda2009forest}, placement of wireless sensors in smart grids \cite{liu2012wireless} and even performance art \cite{carlson2020perceiving}. The fundamental challenge in the control of multi-agent systems arises from the stringent requirements placed on their scalability, communication and privacy. As these requirements cannot be satisfied by a centralized approach, we must use distributed protocols where the agents act independently according to their local information. A common approach for distributed design is to cast the global system objective as an optimization problem where the aforementioned requirements are embedded as constraints \cite{nedic2014distributed}. Then, through careful consideration of the problem's structure, a distributed algorithm is designed that gives good approximate solutions. 

A commonly held assumption in the design of distributed protocols is that all the agents either have perfect knowledge on the underlying problem setting or quickly obtain it through communication (see, e.g., \cite{li2012designing}). However, in practice, an agent's knowledge on the \emph{ground truth} may be limited, especially in scenarios where accurate estimation of the true state of the environment is difficult. This prompts the following line of questioning:

\begin{itemize}[leftmargin=*]
    \setlength\itemsep{.3em}
    \item[\textendash] \emph{How robust is the performance of a given distributed control 
algorithm to inconsistencies in agents' knowledge?}
    \item[\textendash] \emph{Can distributed control algorithms be explicitly designed 
to be robust against inconsistencies in agents' knowledge?}
\end{itemize}

In this work, we investigate these two questions from a game theoretic perspective. Our main contribution is a general framework for evaluating the robustness of distributed control algorithms in which the agents' decision making is based on their own (possibly inaccurate) knowledge of the problem parameters. Each agent acts in its own self interest, maximizing its utility function according to its local knowledge of the underlying problem setting. The system performance guarantees are measured under the well studied notion of \emph{price of anarchy}~\cite{koutsoupias1999worst}, defined as the ratio between the system welfare at the worst system outcome and the system optimum. Here, the worst system outcome is defined as the worst emergent allocation of agents (i.e., pure Nash equilibrium) with respect to the worst possible disposition of the agents' knowledge. We then apply our framework to the class of set covering games~\cite{gairing2009covering} in the setting where each agent's estimate of the problem parameters lies within a bounded interval centered around the true system state.

The study of uncertainty among agents is not new in the field of game theory. In fact, John Harsanyi's seminal work on \emph{incomplete information games}~\cite{harsanyi1967games} was one of the first significant contributions to this field. The incomplete information model has been used to study multi-agent systems in a variety of contexts, including network games \cite{sagduyu2011jamming}, team decision making \cite{billard1999learning} and pursuit/evasion \cite{antoniades2003pursuit}. While this body of important work is relevant to our problem setting, our analysis departs from the incomplete information framework as, in our setting, agents do not possess probabilistic models of the system state and have only limited knowledge on the other agents' beliefs. Thus, the system operator must account for the worst possible scenario when designing the distributed protocols. The differences between the incomplete information framework and the framework proposed in this manuscript are roughly comparable to those between stochastic and robust optimization.

The literature on ``robust'' formulations is much more restricted. In \cite{weber1985method}, the authors consider how to generate a set of possible utility functions that are consistent with a limited amount of information. A distribution-free analysis of incomplete information games is considered in \cite{aghassi2006robust} through a proposed equilibrium concept. To the author's knowledge, the closest work is \cite{meir2015playing}, where price of anarchy results are considered in scenarios where the agents have biases in their perceived utilities. While similar in flavor, our paper studies utility design under a difference class of utility deviations that result from limited information scenarios.

In this paper, we propose a novel game theoretic framework for studying multi-agent systems that does not impose the assumption that agents possess perfect knowledge of the underlying problem setting. We apply our framework to the class of set covering games \cite{gairing2009covering} and study the performance of distributed control algorithms designed without the perfect knowledge assumption in this setting. In Section \ref{sec:model}, we introduce the problem of utility design in uncertain environments and define a relevant class of games in which to model uncertainties. In Section \ref{sec:poa}, we formulate tractable linear programs that (i) tightly characterize performance guarantees and (ii) characterize an optimal utility design in games with system uncertainties, extending the result in \cite{paccagnan2018distributed} to our incomplete information setting. In Section \ref{sec:setcov}, we restrict our attention to the important class of set covering games with uncertainties. In this setting, we obtain an exact characterization of the best achievable efficiency under varying levels of uncertainty. Proofs for all the results are given in the appendix.

\noindent \textbf{Notation.} We use $\R$ and $\N$ to denote the set of real numbers and natural numbers respectively. For any $p, q \in \N$ with $p < q$, let $[p] = \{1, \dots, p \}$ and $[p, q] = \{p, \dots, q \}$. Given a set $S$, $|S|$ represents its cardinality.


\section{The Model}
\label{sec:model}

We expand on the class of resource allocation games \cite{marden2008distributed} to model a general distributed scenario. A set of agents $\ps =[n]$ must be allocated to a set of resources $\rs=\{r_1,\dots,r_m\}$. Each agent $i\in \ps$ is associated with a set of permissible actions $\ac_i \subseteq 2^\rs$ and we denote the set of admissible joint allocations by the tuple $a = (a_1, \dots, a_n) \in \ac = \ac_1 \times \dots \times \ac_n$. Each resource $r \in \rs$ is associated with a state $y_r$ and a welfare function $\wf_r:[n] \times \x_r \to\R$ that measures the system performance at that resource as a function of its aggregate utilization. In other words, $\wf_r(k; y_r)$ is the system performance at resource $r$ when there are $k\in[n]$ users selecting $r$ and the local state of resource $r$ is $y_r$.  Lastly, the system-level welfare is captured by the function $\wf : \ac \times \x \to \R$, where $\x = \prod_{r \in \rs} \x_r$ defines the set of possible system states. In general, for a given allocation $a \in \ac$ and system state $y \in \x$ the system-level welfare is of the form

\begin{equation}
    \wf(a; y) = \sum_{r \in \rs} \wf_r(|a|_r; y_r),
\end{equation}
where $|a|_r = |\{i\in \ps \text{ s.t. }r\in a_i\}|$ represents the number of agents selecting $r$ in the allocation $a$.  Given a system state $y \in \x$, the goal of the multi-agent system is to coordinate to the allocation that optimizes the system-level welfare

\begin{equation} 
    \label{eq:optimal_allocation}
    a^\opt \in \underset{a \in \ac}{\arg \max} \ \wf(a; y).
\end{equation}

However, deriving the optimal allocation requires an infeasible amount of computational resources and coordination. Therefore, we focus on arriving at approximate solutions through a designed utility function $\ui: \ac \times \x \to \R$ and model the emergent collective behavior by a pure Nash equilibrium, which we will henceforth refer to simply as an equilibrium. Given the local knowledge $y_i \in \x$ of each agent $i \in \ps$, an equilibrium is defined as an allocation $a^\ns \in \ac$ such that for any agent $i \in \ps$

\begin{equation}
    \ui(a_{i}^\ns, a_{-i}^\ns; y_i) \geq \ui(a_i, a_{-i}^\ns; y_i), \ \ \forall a_i \in \ac_i.
\end{equation}
It is important to highlight that an equilibrium may or may not exist for such situations, particularly in cases where the agents do not evaluate their utility functions for the same state, i.e., $y_i \neq y_j$. Throughout this paper, we will assume that equilibria exist within the games under consideration. Nonetheless, our results extend to other solution concepts that are guaranteed to exist such as coarse correlated equilibrium \cite{chandan2019smoothness}.

Throughout this paper, we will focus on the scenario where there is a true system state $\xtr \in \x$ and each agent $i \in \ps$ has its own knowledge of this state $y_i$ which may or may not reflect the true state, i.e., $y_i$ need not equal $\xtr$. While the agents' knowledge $y = (y_1, \dots, y_n) \in \x_n$ will invariably influence their local behavior and resulting equilibria, we will measure the performance of the resulting equilibrium $a^\ns$ according to the true state, i.e., $\wf(a^{\ns}; \xtr)$. Accordingly, our goal is to assess how discrepancies in the agents' knowledge impacts the quality of the resulting equilibria. Furthermore, we investigate the optimal design of the utility functions in scenarios where such discrepancies may exist.

To ground these questions moving forward, we consider an extension of the utility functions considered in the framework of Distributed Welfare Games \cite{marden2008distributed}, where each resource is associated with a utility generating function of the form $\ur : [n] \times \x_r \to \R$.  Here, the utility generating function defines the benefit associated with each agent selecting resource $r$, and can depend on both (i) the number of agents selecting resource $r$ and (ii) the state of resource $r$. Given these utility generating functions $\{\ur\}_{r \in \rs}$, the utility of an agent $i\in \ps$ in an allocation $a \in \ac$ is separable and of the form

\begin{equation} \label{eq:utility_function}
    \ui(a; y_i) = \sum_{r\in a_i} \ur(|a|_r; y_{i,r}).
\end{equation}
Note that each agent $i\in\ps$ uses its own state values $y_i = \{y_{i,r}\}_{r \in \rs}$ to compute its utility at each resource. 

We measure the efficiency of the resulting equilibria through the well-studied price of anarchy metric \cite{koutsoupias1999worst}. We begin by formally expressing a game by the tuple 
\begin{equation*}
   \gm = \big(\ps, \rs, \ac, \big\{\wf_r, \x_r, \ur \big\}_{r \in \rs}, \xtr, \big\{y_i \big\}_{i \in \ps} \big).
\end{equation*}
Note that this tuple includes all relevant information to define the game. We define the price of anarchy of the game $\gm$ by

\begin{equation*}
    \poa(\gm) := \frac{\min_{a \in {\rm NE}(\gm)} \wf(a; \xtr)}{\max_{a \in \ac} \wf(a; \xtr)} \leq 1,
\end{equation*}
where ${\rm NE}(\gm) \subseteq \ac$ denotes the set of equilibrium of the game $\gm$. We will often be concerned with characterizing the price of anarchy for broader classes of games where resources share common characteristics. To that end, let $\z_r = \big\{ \wf_r, \x_r, \ur \big\}$ define the characteristics of a given resource $r$. Further, let $\Z$ denote a family of possible resource characteristics. We define the family of games $\Gm_{\Z}$ as all games of the above form where $\big\{\wf_r, \x_r, \ur \big\} \in \Z$ for each resource $r \in \rs$. The price of anarchy of the family of games $\Gm_{\Z}$ is defined as 

\begin{equation*}
    \poa(\Gm_{\Z}) := \inf_{\gm \in \Gm_{\Z}} \poa(\gm) \leq 1.
\end{equation*}

For brevity we do not explicitly highlight the number of agents in a class of games as that is always assumed to be less than $n$. In order to express the informational inconsistencies between the agents' knowledge and the true state, we define the metric $\rho_d: \Gm_{\Z} \to \R_{\geq 0}$ as
\begin{equation}
    \rho_d(\gm) = \max_{i\in\ps} d(y_i^\gm; \xtr^\gm),
\end{equation}
where $d: \x \times \x \to \R_{\geq 0}$ is some distance measure such that $d(y,y')=0$ if and only if $y=y'$. Observe that, under this notation, a perfect information scenario where all agents know the true state corresponds with $d(y_i^\gm; \xtr^\gm) = 0$ for all $i \in \ps$ and $\rho_d(\gm) = 0$. Conversely, when the agents have limited knowledge on $\xtr^{\gm}$, $\rho_d(\gm)$ measures the extent of the uncertainty where a higher $\rho_d(\gm)$ indicates that the agent's evaluation of the state $y_i^\gm$ is ``further'' from the true state $\xtr^\gm$. Consolidating these limitations, the set of games in which $\rho_d(\gm) \leq \delta$ is denoted by $\Gm^\delta_{\Z}$.

In particular, we use the following distance measure for the rest of the paper:
\begin{equation}
    \label{eq:dwfs}
    d(y; \xtr) =  \max_{r \in \rs, k \in [n]} \frac{|\wf_r(k; y_r) - \wf_r(k; y_{\rm{true}, r})|}{\wf_r(k; y_{\rm{true}, r})}.
\end{equation}
Note that, for a given instance, this measure allows us to capture the level uncertainty that agents have on the system welfare independently of the individual resources. For context, we introduce the following application domains.

\begin{ex}[Forest Fire Detection]
\label{ex:forest}
Consider the scenario detailed in \cite{hefeeda2009forest} where a set of unmanned aerial vehicles (UAVs) coordinate to cover a forest region to maximize the detection of a forest fire - modeled as a covering game \cite{gairing2009covering}. The UAVs are the agents in the game and the resource set $\rs$ correspond to a finite partition of the forest region that the UAVs are tasked to cover. Each UAV carries a sensor with a limited sensing range and must select a position to survey (with a corresponding sensing range) - this choice is modeled by an action set $\ac_i$. The state $y_r$ of each resource $r$ corresponds to the risk that a forest fire might emerge in that resource. 
We wish to allocate the UAVs to maximize
\begin{equation}
    \wf(a; \{y_r\}_{r\in\mc{R}}) = \sum_{r\in\cup a_i} y_r,
\end{equation}
which must balance focusing on the high risk areas and covering as much of the forest region as possible.
\end{ex}

\begin{ex}[Weapon-Target Assignment]
Consider the weapon-target assignment problem described in \cite{murphey2000target} where a set of weapons $\ps=\{1,\dots,n\}$ are assigned to a set of targets $\mc{T}$ with the objective of maximizing the expected value of targets engaged. When $k\in\{1,\dots,n\}$ weapons engage a target $t\in\mc{T}$, its expected value is $v_t\cdot[1-(1-p_t)^k]$ where $v_t\geq 0$ is $t$'s associated value and $p_t\in[0,1]$ is $t$'s probability of successful engagement. Based on its range and specifications, each weapon $i\in\ps$ can only engage particular subsets of the targets corresponding with the actions $a_i \in \mc{A}_i\subseteq 2^\mc{T}$. Accordingly, under an allocation of weapons $a=(a_1,\dots,a_n)$, the operator's welfare is measured as
\begin{equation}
    \wf(a; \{(v_t,p_t)\}_{t\in\mc{T}}) = \sum_{t\in\mc{T}} v_t \cdot \big[ 1 - (1 - p_t)^{|a|_t} \big].
\end{equation}
Observe that this scenario can be modeled as a resource allocation game where each weapon is an agent, each target is a resource and each target $t$ has state $y_t = (v_t, p_t)$. 
\end{ex}


Though a resource characteristic is a triplet $\z_r = \big\{\wf_r, \x_r, \ur\big\}$, in many cases only $\wf_r$ and $\x_r$ are inherited from the problem setting, while the utility generating rule $\ur$ is designed. Accordingly, we will often think of the utility generating function at each resource as being derived from $\{\wf_r, \x_r\}$, i.e., $\ur = \Pi(\wf_r, \x_r)$ where $\Pi$ is the \emph{utility mechanism}. Let the set $\Z(\Pi) = \{\wf_r, \x_r, \Pi(\wf_r, \x_r)\}_{r \in \rs}$. 

The main focus of this paper is on determining the utility mechanism that maximizes the price of anarchy, i.e.,
\begin{equation}
    \label{eq:ovprob}
    \Pi^{\rm opt} = \underset{\Pi}{\arg \max} \ \poa(\Gm_{\Z(\Pi)}^\delta).
\end{equation}
Accordingly, one may wish to understand how the achievable performance guarantees are affected by the amount of uncertainty $\delta\geq 0$. In scenarios where the system designer does not know the true value of $\delta$, one may additionally seek to characterize the degradation in performance guarantees for estimates on $\delta$ of varying levels of accuracy. 
In the forthcoming sections, we provide preliminary results along these lines of questioning.



\section{Characterization of PoA}
\label{sec:poa}

Having defined our general model for limited information scenarios, in this section, we concentrate on a specific class of resource characteristics. Doing so allows us to formulate the optimization problem in \eqref{eq:ovprob} as a linear program by leveraging recent results in \cite{paccagnan2018distributed}. Let $\Z_{w, u}$ correspond to a set of resource characteristics,\footnote{The results in this section can be extended to settings where $\wf_r$ and $\ur$ are linear combinations over a set of basis functions pair $\{w^j,u^j\}$, $j=1,\dots,L$, following the results in \cite{chandan2019smoothness}. We state our results for only one basis function pair $\{w,u\}$ (i.e., $L=1$) for ease of presentation.}
\begin{align}
    \x_r &= \R_{\geq 0} \\
    \wf_r(|a|_r; y_r) &= y_r \cdot w(|a|_r) \\
    \ur(|a|_r; y_r') &= y_r' \cdot u(|a|_r)
\end{align}
respectively, where $y_r, y_r' \in \x_r$ and $w: [n] \to \R_{>0}$ and $u: \{1, \dots, n\} \to \R$ are fixed across all resources $r \in \rs$ with $w(0) = u(0) = 0$ and $w(1) = u(1) = 1$. With abuse of notation, we use the denotation $\Gm_{w, u}$ to refer to the family of games $\Gm_{\Z_{w, u}}$. In this model, $y_r$ corresponds to a measure of value of the resource $r$. Additionally, when a set of agents cover a certain resource $r$, $w$ and $u$ correspond to the resource agnostic measure of the added system welfare and agent utility, respectively. In this setting, the distance measure in \eqref{eq:dwfs} can be rewritten as
\begin{align*}
    d(y; \xtr) &=   \max_{r \in \rs, k \in [n]} \frac{| y_r \cdot w(k) - y_{\rm{true}, r} \cdot w(k)|}{ y_{\rm{true}, r} \cdot w(k)} \\
    &= \max_{r \in \rs} \frac{|y_{r} - y_{\rm{true}, r}|}{y_{\text{true}, r}},
\end{align*}
directly encoding the relative uncertainty of $y$ from $\xtr$. In other words, given a maximum uncertainty $0\leq \delta \leq 1$, the state $y_r$ must be in the continuous interval $[(1-\delta)  y_{\rm{true}, r}, (1+\delta) y_{\rm{true}, r}]$ for all $r$.\footnote{Note that we only consider the domain $\delta \in [0, 1)$, since if $\delta \geq 1$, the player valuation $y_r^i$ can be arbitrarily close to $0$ for any $y_{\rm{true}, r}$} In the forthcoming results, we will also use the parameter 
\begin{equation*}
    \cfdel = \frac{1 + \delta}{1 - \delta}
\end{equation*}
to state certain equations more concisely.
%
%
The following theorem presents a tractable linear program for computing the price of anarchy:
\begin{thm}
\label{thm:LP}
Consider a class of resource allocation games with $\Z_{w, u}$ for a given $w$ and $u$. Additionally, let $\delta \in [0, 1)$ denote the limitations of the agents' knowledge. It holds that $\poa(\Gm_{w, u}^{\delta})=1/V^*$ where $V^*$ is the optimal value of the following linear program:
\begin{equation}\label{eq:LP1}
\begin{aligned}
    V^* = & \max_{\theta} \sum_{a, x, b}{w(b+x) \theta(a, x, b)} \ \text{s.t.} \\
    & \sum_{a, x, b}{\Big[\cfdel a u(a+x) - b u(a+x+1)\Big] \theta(a, x, b) \geq 0} \\
    & \sum_{a, x, b}{w(a+x) \theta(a, x, b)} = 1 \\
    & \theta(a, x, b) \geq 0 \ \ \forall a,x,b
\end{aligned}
\end{equation}
where $a, x, b \in \mathbb{N}$ such that $1 \leq a + x + b \leq n$.
\end{thm}

Following the reasoning detailed in \cite{paccagnan2018distributed}, we can also define a linear program that computes the optimal utility design for the class of resource allocation games. Interestingly, we can directly import the techniques in \cite{paccagnan2018distributed} to achieve quite strong answers to questions about utility designs in limited information settings. For a given uncertainty $\delta$ and welfare characteristic $w$, we refer to the optimal utility mechanism as $\opu$. We omit the details of the proof. 

\begin{cor}
\label{cor:optfw}
Consider the class of resource allocation games $\Gm_{w, u}^{\delta}$ with $n$ number of agents for a given $w \in \R^n_{> 0}$. Additionally, let $\delta \in [0, 1)$ denote the uncertainty. The utility mechanism $\opu$ that maximizes the price of anarchy is given as
\begin{align*}
    &(\opu, \mu^*) \in \underset{u \in \R^n,  \ \mu \in \R}{\arg \min} \mu \ \ \ \mathrm{s.t.}  \\
    w(b + x) &- \mu w(a + x) + \cfdel a u(a + x) - b u(a + x + 1) \leq 0 \\
    &\mathrm{for \ all} \ a, x, b \in \mathbb{N} \ \mathrm{with} \  1 \leq a + x + b \leq n \\
    &u(1) = 1
\end{align*}
with $\poa(\Gm_{w, \opu}^{\delta}) = \frac{1}{\mu^*}$.
\end{cor}

In a realistic scenario, the system operator may not know the extent of the informational inconsistencies among the agents (i.e., the precise value of $\delta$). In this case, what can be shown about the performance guarantees that a utility $\opu$ -- designed according to an assumed uncertainty $\delta$ -- achieves under a realized uncertainty $\dtr \neq \delta$? In other words, if there is a mismatch between the operator's assumed $\delta$ and the realized $\delta$, is there any loss of performance? In these situations, the following figure shows the quite surprising fact that underestimating the $\delta$ actually gives better performance guarantees. In Figure \ref{fig:dropoffsub}, we randomly generated $30$ different welfare characteristics $w$, in which $w(j)$ is concave and non-decreasing in $j$. We assume that $\dtr = .3$ and the game has $n=10$ players. For a given $w$, we computed the optimal utility design $\opu = \arg \max_u \poa(\Gm_{w, u}^{\delta})$ for each $\delta$ from $\delta = 0$ to $1$. Then we plot the price of anarchy $\poa(\Gm_{w, \opu}^{\dtr})$ for each design. We see the performance degrades slower to the left of $\dtr$ than to the right. In the next section, we formally capture this trend in the well-studied class of set covering games.

\begin{figure}
    \centering
    \includegraphics[width=0.4\textwidth]{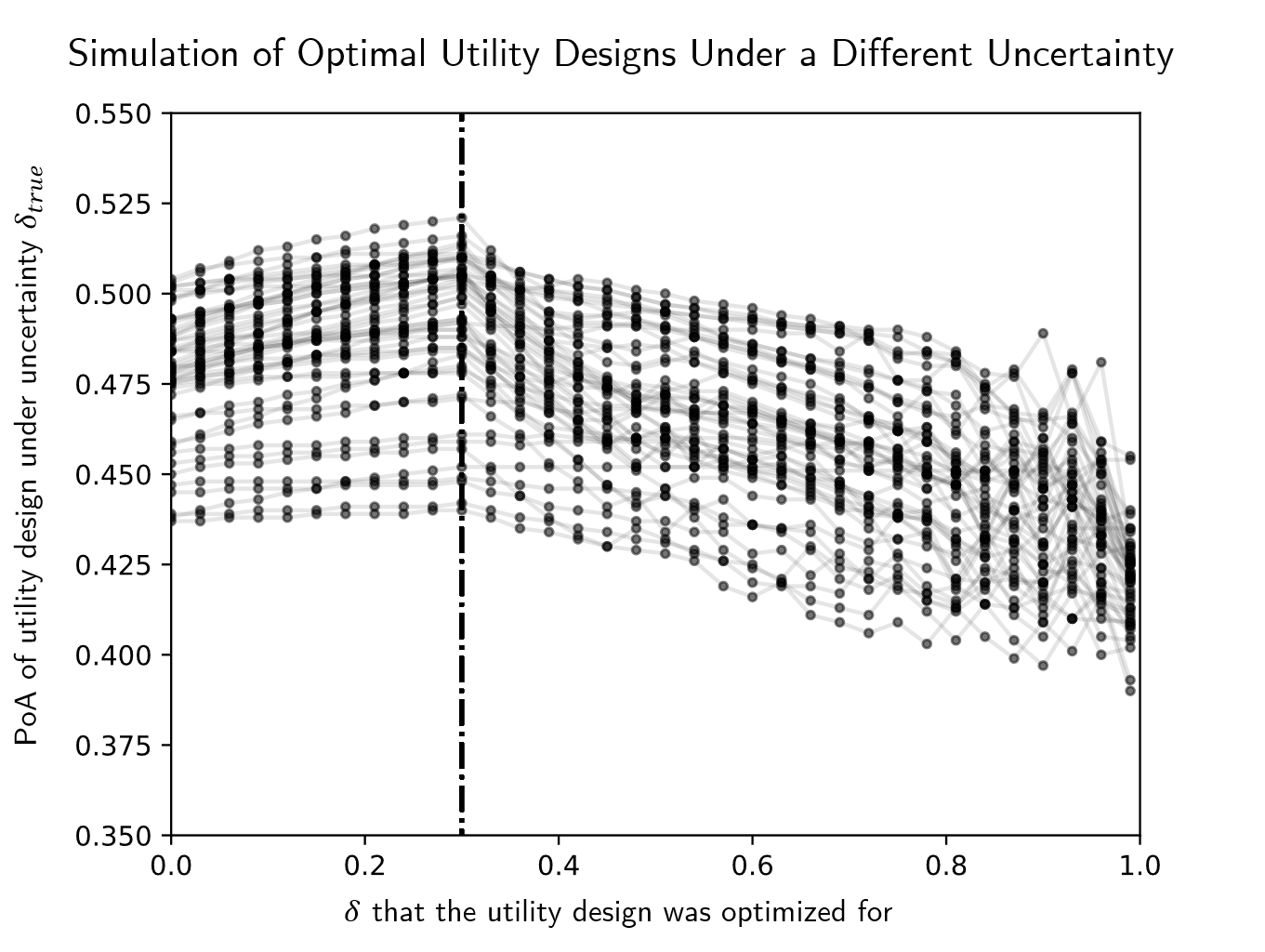}
    \caption{The price of anarchy is plotted for the optimal designs for various uncertainties $\delta$ under 30 randomly chosen welfare characteristics $w$ and given $\dtr = 0.3$ (indicated by the dotted line). Spanning all possible $\delta$, it can be seen indeed that $\pi_{0.3}$ performs optimally. However, more surprisingly, the degradation of performance as we move away from $\dtr$ is slower on the left than on the right of $\delta=0.3$. Nonintuitively, this suggests that by underestimating the value of $\dtr$ we can achieve higher price of anarchy than overestimating it.}
    \label{fig:dropoffsub}
    \vspace{-12pt}
\end{figure}

\section{Set Covering Results}
\label{sec:setcov}

In this section, we restrict our analysis to \emph{set covering games}, where the system welfare is the value of the resources covered, i.e. $w_{sc}(k) = 1$ for $k \geq 1$ and. Set covering games \cite{gairing2009covering} are well studied and known to model a wide variety of practical applications, one of which is detailed in Example \ref{ex:forest}. We obtain an explicit characterization of the optimal utility design $\opu$ for the class of set covering games and formally demonstrate the phenomenon observed in Figure \ref{fig:dropoffsub} for this class of games.

To characterize the optimal utility design, we first outline the following Proposition to characterize the price of anarchy in set covering games with uncertainty.

\begin{prop}
\label{prop:setpoa}
For given utility $u$, uncertainty $0 \leq \delta \leq 1$, and number of agents $n \geq 1$, the price of anarchy for the class of set covering games $\Gm_{w_{sc}, u}^{\delta}$ is
\begin{align*}
    \poa^{-1}(\Gm_{w_{sc}, u}^{\delta}) =
    \max_{j \in [1, n-1]} \{\max \{\cfdel (j+1) u(j+1), \\
    \cfdel j u(j+1) + 1,
     \cfdel j u(j) - u(j+1) + 1 \} \}
\end{align*}
\end{prop}

\begin{thm}
\label{thm:foptsetcov}
For a given $\delta$ the optimal utility design $\opu$ for the class of set covering games is
\begin{equation}
    \label{eq:setcovoptu}
    \opu(j) =  \sum_{k=j}^{\infty} \frac{(j-1)! }{\cfdel^{k - j + 1} (e^\frac{1}{B} - 1) k!}
\end{equation}
and has corresponding price of anarchy
\begin{equation*}
    \poa(\Gm_{w_{sc}, \opu}^{\delta}) = 1- e^{-\frac{1}{\cfdel}}.
\end{equation*} 
\end{thm}

Now that we have characterized the optimal utility design for set covering games, we can arrive at a closed form expression for the guarantees when there is a mismatch in uncertainty between the system operator and the realized uncertainty.

\begin{prop}
\label{prop:setmisalign}
Let $\opu$ be the optimal utility design for $0\leq\delta\leq 1$ as in \eqref{eq:setcovoptu} and $0\leq\dtr\leq 1$ be the realized uncertainty. The price of anarchy is $\poa(\Gm_{w_{sc}, \opu}^{\dtr}) = V^{-1}$
\begin{equation}
\label{eq:opucurve}
V = 
\begin{cases}
       (\cfdeltr \cfdel^{-1} - 1) \opu(2) +  \\ 
        \cfdeltr \cfdel^{-1}(C - 1) + 1  & \text{if } \delta \leq \dtr,\\
        \cfdeltr \cfdeltr^{-1}(C - 1) + 1 & \text{if } \delta \geq \dtr,
        \end{cases}
\end{equation}
where C = $(e^{\frac{1}{\cfdel}} - 1)^{-1}$.
\end{prop}

\begin{figure}
    \centering
    \includegraphics[width=250pt]{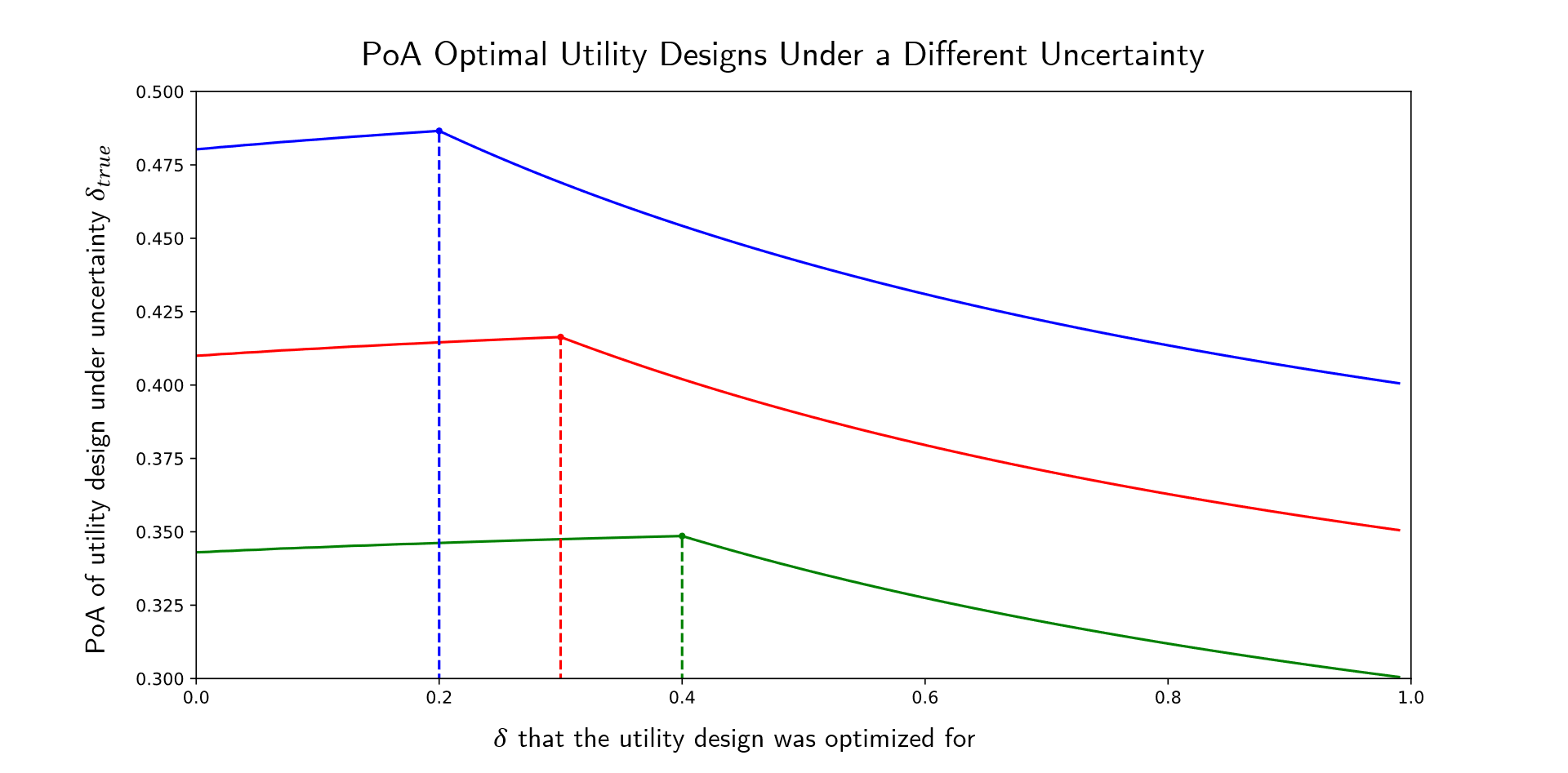}
    \caption{We plot the price of anarchy achieved by the utility rules designed for $\delta\in[0,1]$ within three classes of set covering games corresponding with $\dtr=0.2$ (in blue), $\dtr=0.3$ (in red), and $\dtr=0.4$ (in green). The explicit expression for such curves is provided in \eqref{eq:opucurve}. Observe that underestimating the true level of uncertainty in the class of games gives better price of anarchy guarantees than overestimating it, a trend that we previously noted from the simulation results in Figure \ref{fig:dropoffsub}.}
    \label{fig:dropsetcov}
    \vspace{-12pt}
\end{figure}

In the context of set covering games with uncertainty in the state of the resources, we have formally shown the surprising fact that underestimating these values actually has better performance guarantees than overestimating.

\section{Conclusion}
\label{sec:conc}
In this paper, we consider the problem of designing distributed algorithms for multi-agent scenarios under informational inconsistencies. This is necessary when the true system cannot be known or well-estimated, either due to the computational or communication overhead. We study such incomplete information scenarios in the context of local resource allocation games, and use the price of anarchy to evaluate the quality of the emergent Nash equilibrium of the system. We first outline a tractable linear program to give a tight characterization of the price of anarchy under different levels of system uncertainties as well as a tractable linear program to outline the optimal utility design. We then fully characterize the price of anarchy of the special subclass of set covering games with uncertainty and we examine the effect of uncertainty on optimality of the design utilities. This paper is an important first step to providing design strategies for multi-agents scenarios with more realistic assumptions on the information available to the agents.

\bibliographystyle{ieeetr}
\bibliography{bibliography.bib}

\appendix
\begin{proof}[Proof of Theorem \ref{thm:LP}]
First, we show that $\poa(\Gm_{w, u}^{\delta})$ is lower bounded by $1/V^*$. Consider the reduced family of games $\Gm_{w, u}^{\delta, 2} \subset \Gm_{w, u}^{\delta}$ where all the agents have only two actions; $\ac_i = \{a_i^{\ns}, a_i^{\opt} \}$. In this reduced game, $a^{\ns}$ corresponds to a Nash equilibrium and $a^\opt$ corresponds to the action that maximizes the welfare. Note that $\poa(\Gm_{w, u}^{\delta}) = \poa(\Gm_{w, u}^{\delta, 2})$ and that for any game $G\in\Gm_{w, u}^{\delta}$, uniformly scaling the values $y_r$ such that $\wf(a^{\ns}) = \sum_{r \in \rs}{y_r w(|a^\ns_r|)} = 1$ does not affect the price of anarchy. It follows that $\poa(\Gm_{w, u}^{\delta}) = 1/W^*$ where 
\begin{equation} \label{eq:gform}
\begin{aligned}
    W^* = & \max_{\gm \in \Gm_{w, u}^{\delta, 2}} \ \wf(a^{\opt}) \\
    \text{s.t.} \ & \ui(a^{\ns}; \delta) \geq \ui(a_i^{\opt}, a_{-i}^{\ns}; \delta), \quad  \forall i \in \ps \\
    & \wf(a^{\ns}) = 1.
\end{aligned}
\end{equation}
Observe that, as written, the above linear program is intractable as there are infinitely many games in $\Gm_{w, u}^{\delta, 2}$. To reduce the complexity, we define a game parameterization based on $n$ partitions of the set of resources, defined as follows for each agent $i\in\ps$:
\begin{align*}
    \rs_{a^{\ns}_i}                 & = \{r \in \rs : r \in a^{\ns}_i \setminus a^{\opt}_i\}, \\
    \rs_{a^{\opt}_i}                & = \{r \in \rs : r \in a^{\opt}_i \setminus  a^{\ns}_i\}, \\
    \rs_{a^{\opt}_i \cap a^{\ns}_i} & = \{r \in \rs : r \in a^{\opt}_i \cap a^{\ns}_i\}, \\
    \rs_{a^{\varnothing}_i}         & = \{r \in \rs : r \notin a^{\opt}_i \cup a^{\ns}_i\}.
\end{align*}

Now consider an arbitrary game $\gm \in \Gm_{w, u}^{\delta, 2}$ with resources $\rs$, agent valuations $y_i$ for each agent $i$ and true resource values $\xtr$. We can rewrite the Nash constraint in \eqref{eq:gform} for each agent $i\in\ps$ as
\begin{equation*}
    \sum_{r \in a_i^{\ns}} y^i_r \cdot u(|a^{\ns}_r|) \geq \sum_{r \in a_i^{\opt}} y^i_r \cdot u(|(a_i^{\opt}, a_{-i}^{\ns})_r|).
\end{equation*}
Under the partition defined for each agent $i$, we observe that the Nash condition can be rewritten as
\begin{align*}
    & \sum_{r \in \rs_{a^{\ns}_i}} y^i_r \cdot u(|a^{\ns}_r|) + \sum_{r \in \rs_{a^{\opt}_i \cap a^{\ns}_i}} y^i_r \cdot u(|a^{\ns}_r|)  \\
    \geq \ & \sum_{r \in \rs_{a^{\opt}_i}} y^i_r \cdot u(|(a_i^{\opt}, a_{-i}^{\ns})_r|) \\
    & \quad + \sum_{r \in \rs_{a^{\opt}_i \cap a^{\ns}_i}} y^i_r \cdot u(|(a_i^{\opt}, a_{-i}^{\ns})_r|).
\end{align*}
Canceling the terms in $\rs_{a^{\opt}_i \cap a^{\ns}_i}$, we get
\begin{equation*}
    \sum_{r \in \rs_{a^{\ns}_i}} y^i_r \cdot u(|a^{\ns}_r|) 
    \geq \sum_{r \in \rs_{a^{\opt}_i}} y^i_r \cdot u(|(a_i^{\opt}, a_{-i}^{\ns})_r|).
\end{equation*}
Note that, for any resource $r\in\rs$, it must hold that $y^i_r \in [(1-\delta)  y_{\rm{true}, r}, (1+\delta) y_{\rm{true}, r}]$. Considering the Nash condition as written above, observe that the tightest constraint arises when $y^i_r = (1+\delta) y_{\rm{true}, r}$ for all $r \in \rs_{a^{\ns}_i}$, $y^i_r = (1-\delta) y_{\rm{true}, r}$ for all $r \in \rs_{a^{\opt}_i}$ and $y^i_r = y_{\rm{true}, r}$ for all other resources. This is the situation where agents overvalue the resources in their equilibrium actions and undervalue the resources in their optimal actions. Thus, we can consider this situation without loss of generality. For each resource $r\in\rs$, we define the triplet $(a_r,x_r,b_r)\in\N^3$ as $a_r = |\{i\in\ps : r \in \rs_{a^{\ns}_i}\}|$, $b_r = |\{i\in\ps : r \in \rs_{a^{\opt}_i}\}|$, and $x_r=|\{i\in\ps : r \in \rs_{a^{\opt}_i\cap a^{\ns}_i}\}|$ where $1\leq a_r+x_r+b_r\leq n$ must hold. Further, we define the map $\theta:\N^3\to\R$ such that $\theta(a,x,b)$ is equal to the sum over true values $y_{{\rm true},r}$ for all resources with $a_r=a$, $x_r=x$ and $b_r=b$, for all $a,x,b\in\N$ with $1\leq a+x+b\leq n$. Under this notation, the following expressions hold:
\begin{align*}
    \wf(a^{\opt}) & = \sum_{a,x,b} w(b+x)\theta(a,x,b), \\
    \wf(a^{\ns})  & = \sum_{a,x,b} w(a+x)\theta(a,x,b).
\end{align*}
We showed above that the tightest Nash condition arises when agents overvalue the resources that they select in their equilibrium actions and undervalue resources in their optimal actions. Under our parameterization, the sum over agents' utilities in this ``tightest'' scenario are as follows:
\begin{align*}
    \sum^n_{i=1} \ui(a^{\ns};\delta)  & = \sum_{a,x,b} [(1+\delta)a + x]  u(a+x) \theta(a,x,b), \\
    \sum^n_{i=1} \ui(a^{\opt}_i,a^{\ns}_{-i};\delta) & 
                = \sum_{a,x,b} (1-\delta)bu(a+x+1)\theta(a,x,b) \\
                    & \qquad + \sum_{a,x,b} xu(a+x) \theta(a,x,b).
\end{align*}
Observe that if the equilibrium condition in \eqref{eq:gform} holds then the sum over agents' utilities at equilibrium must be greater than or equal to the sum over each of their utilities after they unilaterally deviate. The converse, however, need not hold in general. Thus, the linear program \eqref{eq:LP1} in the claim represents a relaxation of the linear program \eqref{eq:gform}. 

Let $V^*$ and $W^*$ be the optimal values of linear programs \eqref{eq:LP1} and \eqref{eq:gform}, respectively. According to the proof thus far, we can only say that $V^*\geq W^*$ since $V^*$ is the optimal value of a relaxed linear program, which means that $\poa(\Gm_{w, u}^{\delta})\geq 1/V^*$. To show that $\poa(\Gm_{w, u}^{\delta})\leq 1/V^*$ also holds, one can follow the approach outlined in \cite{paccagnan2018distributed}, which we omit here due to space constraints.
\end{proof}

\begin{proof}[Proof of Proposition \ref{prop:setpoa}]
This proof is inspired by Theorem~3 in \cite{paccagnan2021utility} with added consideration for the informational inconsistencies between the agents, and is included for completeness. First, we write the Lagrange dual of the linear program \eqref{eq:LP1}, i.e., $\poa(\Gm_{w_{sc}, u}^{\delta})=1/\mu^*$ where
{ \small
\begin{align*}
    \mu^* & = \min_{\lambda \geq 0,  \ \mu \in \R} \mu \ \mathrm{s.t.} \  \\
    & \mu w(a + x) \geq w(b + x) + \lambda \big[\cfdel a u(a + x)  - b u(a + x + 1)\big], \\
    & \hspace*{90pt} \forall a, x, b \in \mathbb{N} \ \mathrm{s.t.} \  1 \leq a + x + b \leq n.
\end{align*}
}%
The rest of the proof involves removing redundant constraints to obtain a closed form expression of the price of anarchy. We first consider the constraints that arise from $a = x = 0$ and $b \geq 1$. By definition, $w(k) = 1$ if $k \geq 1$ and $0$ otherwise, giving the constraint $\lambda \geq \max_{b \in [n]} \frac{1}{b} = 1$. Considering the set of constraints that arise from $b, x = 0$ and $a \geq 1$ gives $\mu \geq \lambda \cfdel a u(a) \ \mathrm{for} \ a \in [n]$. Now evaluating the set of constraints that arise from $x = 0$, $a, b \geq 1$ gives
\begin{align*}
    \mu &\geq \max_{a + b \in [2, n]} 1 + \lambda\big[\cfdel a u(a)  - b u(a + 1) \big] \\
    &\geq \max_{a \in [1, n-1]} 1 + \lambda\big[\cfdel a u(a)  - u(a + 1) \big],
\end{align*}
which holds since $b = 1$ is the most binding constraint. When we consider the constraints that arise from $a, x, b \geq 1$, the resulting set of constraints can be written as
\begin{align*}
    \mu &\geq \max_{a + b + x \in [3, n]} 1 + \lambda\big[\cfdel a u(a + x)  - b u(a + x + 1) \big] \\
    &\geq \max_{a + x \in [2, n]} 1 + \lambda \cfdel a u(a + x) \\
    &\geq \max_{a \in [1, n-1]} 1 + \lambda \cfdel a u(a + 1),
\end{align*}
where $b = 0$ and $x = 1$ is the most binding constraint. Setting $b = 0$ is binding since it removes the negative term $- b u(a + x + 1)$ from the expression. Setting $x = 1$ is binding, since for any pair $\{a, x\} \in [2, n]$, there is another pair $\{a+x-1, 1\}$ that results in stricter constraint. With the nonbinding constraints removed, the program reduces to
\begin{align*}
    &\min_{\lambda \geq 1,  \ \mu \in \R} \mu \ \mathrm{s.t.} \  \\
      \mu &\geq \lambda \cfdel a u(a) &a \in [n] \\
      \mu &\geq 1 + \lambda\big[\cfdel a u(a)  - u(a + 1) \big]] &a \in [n-1] \\
      \mu &\geq 1 + \lambda \cfdel a u(a + 1)  &a \in [n-1]
\end{align*}
The optimal dual variables have $\lambda = 1$, which comes from the tightest constraints. If we assume $a = 1$, this results in the set of constraints
\begin{align*}
    \mu &\geq \cfdel 1 u(1)  = \cfdel \\
    \mu &\geq 1 + \cfdel 1 u(1)  - u(2)  = \cfdel + (1 - u(2)) \\
    \mu &\geq 1 + \cfdel 1 u(2)  = 1 + \cfdel u(2).\\
\end{align*}
We can see the first constraint is always redundant, no matter if $u(2) \geq 1$ or $u(2) \leq 1$. The expression in the claim follows after removing the last nonbinding constraint and shifting the index from $a$ to $a+1$ for the first set of constraints.
\end{proof}

\begin{proof}[Proof of Theorem \ref{thm:foptsetcov}]
First we show that the price of anarchy is lower bounded by the proposed formula  with the corresponding proposed $f$. We assume that the number of agents is $n$. From Proposition \ref{prop:setpoa}, we have that 
\begin{align*}
    &\poa^{-1}(\Gm_{w_{sc}, u}^{\delta}) \leq \mc{X}  \mathrm{\ for \ any \ } \mathcal{X} \mathrm{\ s.t. \ } \\
    \mathcal{X} &\geq \cfdel (n-1) u(n) + 1, \\
        \mathcal{X} &\geq \cfdel j u(j) - u(j+1) + 1 \quad j \in [1, n-1]
\end{align*}
where we removed the first set of constraints, and all but the last one of the second constraints. An optimal utility design satisfies the set of inequalities with equality as follows
\begin{align}
    \label{eq:optchar1}
    \mathcal{X} &= \cfdel (n-1) \opu(n) + 1 \\
    \label{eq:optchar2}
    \mathcal{X} &= \cfdel j \opu(j) - \opu(j+1) + 1 \quad j \in [1, n-1].
\end{align}
We can reformulate this system of equations as a recursive formula to generate the optimal utility design $\opu$ as follows
\begin{align}
    \label{eq:optrec1}
    \opu(n) & = 1 \\
    \label{eq:optrec2}
    \opu(j) & = \frac{\opu(j+1)}{\cfdel j} + \frac{1}{j} \opu(n)(n-1).
\end{align}
Iterating through this recursive equation and normalizing so that $\opu(1) = 1$ gives 
{\small
\begin{equation}
    \label{eq:setcovoptn}
    \opu(j) = \frac{\cfdel^{j-1} (j-1)! (\frac{1}{\cfdel^n (n-1)(n-1)!} + \sum_{k=j}^{n-1}{\frac{\cfdel^{-k}}{k!}})}{\frac{1}{\cfdel^n (n-1)(n-1)!} + \sum_{k=1}^{n-1}{\frac{\cfdel^{-k}}{k!}}},
\end{equation}
}%
with a corresponding price of anarchy expression of 
\begin{equation*}
    \poa(\Gm_{w_{sc}, \opu}^{\delta}) \geq 1 - \frac{1}{\frac{1}{\cfdel^n (n-1)(n-1)!} + \sum_{k=0}^{n-1}{\frac{\cfdel^{-k}}{k!}}}. 
\end{equation*}

Taking the limit as $n \to \infty$ and using the identity $\sum_{k=0}^{\infty}{\frac{\cfdel^{-k}}{k!}} = e^{\frac{1}{\cfdel}}$, we observe that $\opu$ corresponds to the expression in \eqref{eq:setcovoptu} and $\poa(\Gm_{w_{sc}, \opu}^{\delta}) \geq 1- e^{-\frac{1}{\cfdel}}$. 

For the upper bound, we construct an $n$ agent worst case set covering game $\gm^*$ inspired by \cite{gairing2009covering}. All agents have two actions with $\ac_i = \{a_i^\ns, a_i^\opt\}$, coinciding with their equilibrium and optimal actions. To state the allocations of $a^\ns$ and $a^\opt$ concisely, we specify each resource with unique label $\ell: \rs \to 2^n$ as follows. First we partition the resources into $n+1$ groups, $\{\rs_0, \dots, \rs_n\}$. The true value of each resource $r \in \rs_k$ is $y_{\rm{true}, r} = (\cfdel)^k$. There is one resource $r_0 \in \rs_0$ with $\ell(r_0) = \{1\}$. For $k \geq 1$, the set of labels of the resources in $\rs_k$ is exactly the set $[2, n] \times \Perm{n-1}{k-1}$, i.e., the set of permutations without $\{1\}$ as the first element. Therefore, there are $(n-1)\frac{(n-1)!}{(n - k)!}$ in $\rs_k$. For any resource $r \in \rs_k$ with $k \geq 1$, the last element of the label $\ell(r)$ denotes which agent selects the resource $r$ in $a^\opt$ and $\{j \in \ps: j \notin \ell(r)\}$ denotes the set of agents that select the resource $r$ in $a^\ns$. For the resource $r_0 \in \rs_0$, agent $1$ selects it in $a^\opt$, and every agent selects it in $a^\ns$. For example, if the label for the resource $r$ is $\ell(r) = \{2, 3\}$ for a game with $4$ agents, then it must be that $r\in\rs_k$ with $y_{\rm{true}, r} = (\cfdel)^2$. Furthermore, agent $3$ selects $r$ in $a^\opt$, while agents $1$ and $4$ select $r$ in $a^\ns$. 

For any resource $r \in \rs_k$ in $\gm^*$, $n-k$ agents select $r$ in $a^\ns$. Furthermore, for any agent $i \geq 2$ and $k \leq n$, the number of resources in $\rs_k$ that are selected in $a_i^\opt$ (denoted as $|\rs^\opt_{i, k}|$) and the number of resources in $\rs_{k-1}$ that are selected in $a_i^\ns$ (denoted as $|\rs^\ns_{i, k-1}|$) are equal. For agent $1$ and $k \geq 2$, it holds that $|\rs^\opt_{1, k}| = |\rs^\ns_{1, k-1}|$. However, it is important to note that for agent $1$, $|\rs^\opt_{1, 1}| = 0$ while $|\rs^\opt_{1, 0}| = |\rs^\ns_{1, 0}| = 1$. 

The agent valuations $y_{i, r}$ for the resources in $\gm^*$ are as follows for a fixed $\delta$ uncertainty. If $r \in a_i^\ns$, then agent $i$ overvalues it to the extreme where $y_{i, r} = (1 + \delta) y_{\rm{true}, r}$ and if $r \in a_i^\opt$, then agent $i$ undervalues it to the extreme, where $y_{i, r} = (1 - \delta) y_{\rm{true}, r}$. The only exception to this is for agent $1$ and the resource $r_0 \in \rs_0$ where $y_{i, r_0} = y_{\rm{true}, r_0}$ since it is selected in both the optimal and equilibrium allocations by agent $1$.

Now we can verify $a^\ns$ is indeed an equilibrium allocation. For any agent $i\in\ps$, we have
\begin{align*}
    \ui(a^\ns; y_i) &= \sum_{k = 0}^{n} \sum_{r \in \rs^\ns_{i, k}} y_r^i u(n - k)\\
    &= \sum_{k = 0}^{n} \sum_{r \in \rs^\ns_{i, k}} (1 + \delta)  y_{\rm{true}, r} u(n - k) \\
    &= \sum_{k = 0}^{n} |\rs^\ns_{i, k}| (1 + \delta) (\cfdel)^k u(n - k) \\
\end{align*}
\begin{align*}
    &= \sum_{k = 0}^{n} |\rs^\opt_{i, k + 1}| (1 - \delta) (\cfdel)^{k+1} u(n - (k + 1) + 1) \\
    &= \sum_{k = 0}^{n} |\rs^\opt_{i, k}| (1 - \delta)  y_{\rm{true}, r} u(n - k + 1) \\
    &= \ui(a_i^\opt, a_{-i}^\ns; y_i),
\end{align*}

where we take advantage of the fact that no resources in $\rs_{n}$ are selected in $a^\ns$ (i.e., $|\rs^\ns_{i, n}| = 0$) in the fourth equality and that no agents $i \geq 1$ select the resource $r_0 \in \rs_0$ in their optimal allocation (i.e., $|\rs^\opt_{i, 0}| = 0$) for the fifth equality. We can use a similar argument for agent $1$ with additional care taken for the resources in $\rs_0$ and $\rs_1$. Its important to note that under any utility design $u$, the action $a^\ns$ is still an equilibrium and the allocations $a^\ns$ and $a^\opt$ do not change.

Under allocation $a^\ns$ in $\gm^*$, all resources in $\rs_k$ for $k \leq n-1$ are covered while, under the optimal allocation, all resources are covered. We can explicitly write the welfare at both allocations as
{\small
\begin{align*}
    W(a^{\ns}) = \sum_{r \in \rs}{y_{\rm{true}, r} w(|a^\ns_r|)} = 1+ \sum_{k = 1}^{n-1} (n-1) \frac{\cfdel^k (n-1)!}{(n-k)!} \\
    W(a^{\opt}) = \sum_{r \in \rs}{y_{\rm{true}, r} w(|a^\opt_r|)} = 1 + \sum_{k = 1}^{n} (n-1) \frac{\cfdel^k (n-1)!}{(n-k)!}
\end{align*}
}%
Therefore, a lower bound on the price of anarchy is
\begin{equation*}
    \poa(\gm^*) \geq \frac{W(a^{\ns})}{W(a^{\opt})} = 1 - \frac{1}{\frac{1}{\cfdel^n (n-1)(n-1)!} + \sum_{k=0}^{n-1}{\frac{\cfdel^{-k}}{k!}}}.
\end{equation*}

Earlier, we showed that $\poa(\Gm_{w_{sc}, u}^{\delta}) \leq \poa(\gm^*)$ for any utility design. We just showed that $\poa(\Gm_{w_{sc}, \opu}^{\delta}) \geq \poa(\gm^*)$. It follows that the utility $\opu$ defined in \eqref{eq:setcovoptn} is optimal. Furthermore, taking the limit as $n \to \infty$ gives $\poa(\Gm_{w_{sc}, \opu}^{\delta}) \leq \poa(\gm^*) = 1- e^{-\frac{1}{\cfdel}}$.
\end{proof}

\begin{proof}[Proof of Proposition \ref{prop:setmisalign}]
We first assume there are $n$ agents and note that as $\delta \to 1$, the recursive formula in \eqref{eq:optrec1} and \eqref{eq:optrec2} for the optimal utility design, normalized to $\opu(1) = 1$, gives $\opu(j) = \frac{1}{j}$ for $j =  1, \dots, n-1$ and $\opu(n) = \frac{1}{n-1}$. Additionally, observe that as $\delta$ increases, $\opu(j)$ increases for any $j$, since the recursive formula in \eqref{eq:optrec2} produces a slower increasing sequence for a higher $\delta$, so normalizing to $\opu(1) = 1$ gives a larger $\opu(j)$. Thus $\opu(j) \leq \frac{1}{j}$ for $j =  1, \dots, n-1$ for any $\delta$. Note that based on the recursive formula in \eqref{eq:optrec2}, $\opu(j)$ is decreasing in $j$ for any $\delta$. By Proposition \ref{prop:setpoa}, the $\poa(\Gm_{w_{sc}, \opu}^{\dtr})^{-1} = \mathcal{X}$ where $\mc{X}$ is the lowest value satisfying
{\small
\begin{align}
    \label{eq:setinequal1}
    \mathcal{X} &\geq \cfdeltr (j+1) \opu(j+1) &j \in [1, n-1]\\
    \label{eq:setinequal2}
    \mathcal{X} &\geq \cfdeltr j \opu(j+1) + 1 &j \in [1, n-1]\\
    \label{eq:setinequal3}
    \mathcal{X} &\geq \cfdeltr j \opu(j) - \opu(j+1) + 1 &j \in [1, n-1]
\end{align}
}%
Now the redundant inequalities are eliminated to derive a closed form expression. For the  inequalities in \eqref{eq:setinequal1}, we have that 
{\small
\begin{align*}
    \cfdeltr (j+1) \opu(j+1) &\leq \cfdeltr \\
    &\leq \cfdeltr - \opu(2) + 1 \ \  j \in [1,  n-2],
\end{align*}
}%
where the first inequality comes from $\opu(j+1) \leq 1/(j+1)$ and the second inequality comes from $\opu(2) \leq \opu(1) = 1$. Note that putting $j=1$ in the last set of inequalities \eqref{eq:setinequal3} gives the last term. For $j = n-1$, 
\begin{align}
    \cfdeltr n \opu(n) &= \cfdeltr (n-1) \opu(n) + \cfdeltr \opu(n) \nonumber \\
    &\leq \cfdeltr (n-1) \opu(n) + 1 \nonumber \\
    &\leq \cfdeltr \cfdel^{-1}(C - 1) + 1, \label{eq:cons1diffd}
\end{align} 
where the first inequality comes from the fact that $\cfdeltr \opu(n) \leq n$ for a high enough $n$ and the second inequality comes from the substitution $C - 1 = \cfdel (n-1) \opu(n)$ from Equation \eqref{eq:optchar1}.  Note that this corresponds to putting $j = n -1$ in the second set of inequalities  \eqref{eq:setinequal2}. Therefore, we have shown the first set of inequalities is redundant.

For the inequalities in \eqref{eq:setinequal2}, we have that for $j \in [1, n-2]$,
\begin{align*}
    \cfdeltr j \opu(j+1) + 1 &= \cfdeltr(j+1)\opu(j+1) \\
    &- \cfdeltr \opu(j+1) + 1 \\
    &\leq \cfdeltr(j+1)\opu(j+1) \\
    &- \opu(j+2) + 1,
\end{align*}
where the first inequality comes from the fact that $\cfdeltr \opu(j+1) \geq 1 \cdot \opu(j+2)$. Note that this expression matches the inequalities in \eqref{eq:setinequal3} for $j \in [1, n-2]$ and therefore are redundant.

We can also reduce the inequalities in \eqref{eq:setinequal3}:
\begin{align*}
    \cfdeltr j \opu(j) - \opu(j+1) + 1 = \\
    (\cfdeltr \cfdel^{-1} - 1) \opu(j + 1) +  (\cfdeltr \cfdel^{-1})(C - 1) + 1 
\end{align*}
where $C = \poa(\Gm_{w_{sc}, \opu}^{\delta})^{-1}$. The equality comes from the recursive formula in \eqref{eq:optchar2} with substitution $\cfdel j \opu(j) = \opu(j+1) + C - 1$. If $\delta \leq \dtr$, then $\cfdeltr \cfdel^{-1} - 1 \geq 0$ and the binding constraint comes from taking $j = 1$,
\begin{equation*}
    (\cfdeltr \cfdel^{-1} - 1) \opu(2) +  (\cfdeltr \cfdel^{-1})(C - 1) + 1.
\end{equation*}
Conversely if $\delta \geq \dtr$, the binding constraint comes from $j = n -1$,
\begin{align*}
    &(\cfdeltr \cfdel^{-1} - 1) \opu(n) +  (\cfdeltr \cfdel^{-1})(C - 1) + 1 \\
    &\leq \cfdeltr \cfdel^{-1}(C - 1) + 1,
\end{align*}
where the inequality comes from $(\cfdeltr \cfdel^{-1} - 1) \leq 0$ for $\delta \geq \dtr$. Note that this constraint is subsumed by the one in \eqref{eq:cons1diffd}.

Finally, the resulting set of inequalities is
\begin{align*}
    \mathcal{X} &\geq \cfdeltr \cfdeltr^{-1}(C - 1) + 1\\
    \mathcal{X} &\geq (\cfdeltr \cfdel^{-1} - 1) \opu(2) +  (\cfdeltr \cfdel^{-1})(C - 1) + 1.
\end{align*}
We showed that the first constraint is strictest if $\delta \geq \dtr$ and that the second constraint is strictest when $\delta < \dtr$. Taking $n \to \infty$, we have from Theorem \ref{thm:foptsetcov} that $C - 1 = \poa(\Gm_{w_{sc}, \opu}^{\delta})^{-1} - 1 = (1- e^{-\frac{1}{\cfdel}})^{-1} - 1 = (e^{\frac{1}{\cfdel}} - 1)^{-1}$.
\end{proof}

\end{document}